\documentclass{article}

\usepackage[margin=1in]{geometry} 
\usepackage{bm}
\usepackage{tikz}
\usepackage{amsmath}
\usepackage{subcaption}
\usepackage{amssymb}
\usepackage{mathtools}
\usepackage{bbm}
\usepackage{empheq}
\usepackage{wrapfig}
\usepackage{mathtools}
\usepackage{amsthm}
\usepackage{titlesec}
\usepackage{amssymb}
\usepackage{algpseudocode}
\usepackage{hyperref}
\usepackage{enumitem}
\usepackage{float}
\titleformat{\section}{\normalfont\large\bfseries}{\thesection}{1em}{}
\titleformat{\subsection}{\normalfont\normalsize\bfseries}{\thesubsection}{1em}{}
\titleformat{\subsubsection}{\normalfont\normalsize\bfseries}{\thesubsubsection}{1em}{}
\usepackage{authblk}
 \usepackage{booktabs}
\usepackage{amsthm}
\usepackage{xcolor}
\usepackage{colortbl} 
\usepackage{algorithm2e}

\newtheorem{theorem}{Theorem}

\title{\textbf{Improving Surrogate Model Robustness to Perturbations for Dynamical Systems Through Machine Learning and Data Assimilation}}

\author[]{Abhishek Ajayakumar}
\author[]{Soumyendu Raha}

\affil[]{Department of Computational and Data Sciences, Indian Institute of Science, Bangalore 560012, India}

\date{} 

\begin{document}

\maketitle

\begin{abstract}
\textcolor{black}{Many real-world systems are modelled using complex ordinary differential equations (ODEs).} However, the dimensionality of these systems can make them challenging to analyze. Dimensionality reduction techniques like Proper Orthogonal Decomposition (POD) can be used in such cases. However, these reduced order models are susceptible to perturbations in the input. \textcolor{black}{We propose a novel framework that combines machine learning and data assimilation techniques to improving surrogate models to handle perturbations in input data effectively.} \textcolor{black}{Through rigorous experiments on dynamical systems modelled on graphs, we demonstrate that our framework substantially improves the accuracy of surrogate models under input perturbations. Furthermore, we evaluate the framework’s efficacy on alternative surrogate models, including neural ODEs, and the empirical results consistently show enhanced performance.}

\end{abstract}
\section{Introduction}

Complex networks provide valuable insights into real-world systems, such as the spread of epidemics \cite{Brockmann_Helbing_2013} and the understanding of biochemical and neuronal processes \cite{Maslov_Ispolatov_2007}. However, modeling these systems becomes computationally complex when dealing with large state spaces. To address this challenge, methods like Proper Orthogonal Decomposition (POD) (\cite{rathinam_petzold_2003}, \cite{articlePODfield1}, \cite{articlePODfield2}, \cite{articlePODfield3}) are employed to capture patterns while operating in reduced dimensions. \textcolor{black}{While dimensionality reduction techniques like POD are effective, their sensitivity to data necessitates robust surrogate models that can maintain accuracy under input perturbations.} \textcolor{black}{Consequently, the question arises regarding how to enhance a surrogate model to accommodate perturbed data sets. This study proposes a framework designed to improve the robustness of surrogate models, ensuring their reliability under varying input conditions.}

\textcolor{black}{
Surrogate models serve as approximations of complex real-world systems, which are often modeled using ordinary differential equations (ODEs) or partial differential equations (PDEs). In many cases, the complexity of these systems makes it difficult to fully capture their underlying processes. To address the challenges of modeling high-dimensional and chaotic systems, several machine learning (ML) approaches have been proposed to construct surrogate models that provide accurate and computationally efficient approximations. 
Several ML-based approaches have been developed to construct surrogate models that provide accurate and computationally efficient approximations. These include analog models \cite{Lguensat_Tandeo_Ailliot_Pulido_Fablet_2017}, recurrent neural networks (RNNs) \cite{Park, park2010time}, residual neural networks that approximate resolvents \cite{brajard2020combining, dueben2018challenges}, and differential equation-based models such as in \cite{fablet2018bilinear, long2018pde, bocquet2019data, bocquet2020bayesian}.
}A notable study \cite{brajard2020combining} integrates machine learning and data assimilation (DA) to improve predictions in scenarios with sparse observations. However, the DA step is computationally expensive, limiting its practical applicability. To address model inaccuracies, researchers have explored alternative approaches such as weak-constraint DA methods \cite{sakov2018iterative} and ML-DA frameworks \cite{Farchi_2021}.

    \par \textbf{\emph{The contributions of our paper can be summarized as follows:}} 
    \textcolor{black}{
    \begin{enumerate}
        \item We introduce a novel framework (Figure~\ref{Methodology}) that improves the robustness of surrogate models against input perturbations by integrating data assimilation and machine learning techniques.
        \item We conduct extensive experiments on dynamical systems represented as graphs, specifically the diffusion system ($\mathbb{D}$) and the chemical Brusselator model ($\mathbb{C}$) discussed in Section~\ref{Background}.
        For dynamical systems on graphs, we propose a dynamic optimization step, described in Sections ~\ref{dynamiclinear} and ~\ref{dynamoptimrd}, to obtain a sparse graph and reduce memory complexity.  
        \begin{enumerate}
            \item
            For the diffusion equation ($\mathbb{D}$) on graphs (Section~\ref{spatio-temporal}), we derive conditions (Theorem~\ref{connec_theorem}) to guide the constraints in our dynamic optimization step.
            \item  Section~\ref{dynamoptimrd} presents the dynamic optimization step for reaction-diffusion systems on graphs, detailing the required constraints.
        \end{enumerate}
        \item Section~\ref{numerical_results} presents experimental results for the framework on both linear and non-linear dynamical systems represented on graphs, using the POD-based surrogate model. Empirical results in Table~\ref{tab:rmse_comparison} demonstrate that our framework enhances surrogate model performance under input perturbations.
        \item In Section~\ref{benchmark}, we demonstrate how our framework can be integrated into a general setting to reinforce neural ODE-based surrogate models trained on data-driven systems. The improved models exhibit enhanced performance and robustness against input perturbations, outperforming both POD-based models and the original neural ODE solutions, as shown in Figures~\ref{linearplot}, \ref{rd_fig}, and \ref{allcomparenode}. 
    \end{enumerate}
    } 

\section{Background}
\label{Background}
We present the fundamental tools and techniques essential for understanding the methodologies discussed in this paper.
\subsection{Orthogonal collocation method}
The orthogonal collocation method is a well-established numerical technique for solving differential equations, particularly in dynamic optimization problems. It provides a versatile approach applicable to both ordinary and partial differential equations \cite{HEDENGREN2014133, nlp_siam}.
\par \textcolor{black}{The proposed method entails partitioning the problem domain into a discrete set of collocation points, ensuring that the polynomial approximation of the vector field $f(t, z(t))$ satisfies the orthogonality condition. Common choices for collocation points include Legendre, Chebyshev, and Gaussian points.}
\par \textcolor{black}{The method enforces the differential equations at the collocation points, resulting in a system of algebraic equations that can be solved using numerical techniques such as Newton’s method or direct solvers.} \textcolor{black}{The choice of collocation points and basis functions significantly impacts the method's accuracy and efficiency.} To illustrate the method, we provide an example using the Lotka-Volterra system.

\begin{align*}
    \frac{dx}{dt} &= \alpha x - \beta xy \\
    \frac{dy}{dt} &= \delta x y - \gamma y.
\end{align*}
Here $x$ and $y$ denote the prey and predator population. The parameters $\alpha, \beta$ determine the prey growth
rate and $\delta, \gamma$ determine the predator growth rate. \textcolor{black}{If we consider a single collocation element with two nodes per element, the orthogonal collocation method \cite{HEDENGREN2014133} solves the following system of equations to determine the prey and predator populations at time $( t_s )$, denoted as $( x_{t_s}, y_{t_s} )$.}

$t_s N_{2\times 2} \left(\begin{array}{c}
\alpha x_{t_s} - \beta x_{t_s} y_{t_s}  \\
\delta x_{t_s} y_{t_s} - \gamma y_{t_s} 
\end{array} \right) = \left(\begin{array}{c}
x_{t_s}  \\
y_{t_s} 
\end{array} \right) - \left(\begin{array}{c}
x_{t_0}  \\
y_{t_0} 
\end{array} \right)$. $N_{2\times2} = \left(\begin{array}{cc}
0.75 & -0.25 \\
1 & 0
\end{array}\right).$ This method is applied in the dynamic optimization step of our framework (Section \ref{framework}).

\subsection{Spatio temporal propagation in graphs}
\label{spatio-temporal}
Spatiotemporal propagation in complex networks has been widely studied across disciplines such as physics, biology, and neuroscience. Complex networks comprise interconnected nodes whose interactions give rise to dynamic processes. Understanding the mechanisms governing information, signal, or dynamic propagation within these networks is crucial for analyzing their complex behaviors and emergent properties.
\par Spatiotemporal propagation describes the transmission or diffusion of information, signals, or dynamics across interconnected nodes and edges in a complex network. It examines how localized events, disturbances, or modifications in one region of the network evolve and influence other nodes or areas. This phenomenon has significant implications across diverse fields, including physics, biology, and neuroscience. The following examples highlight key applications of complex networks.
\begin{enumerate}
    \item \textbf{Diffusion Processes\; $(\mathbb{D})$:} \textcolor{black}{The study of diffusion phenomena, including heat transfer and molecular diffusion, aids in modeling and optimizing processes involving transport and dispersion. By considering the normalized Laplacian matrix of the graph, the heat equation can be represented as an equivalent discrete dynamical system. Normalized Laplacian matrix $\mathcal{L} = D^{-1/2} L D^{-1/2}$, where $D$ is the diagonal matrix of degrees and $L = D - A$ is the Laplacian matrix of the graph.}
    $$
    \frac{dF}{dt} = -\mathcal{L} F, \;\; 
    $$
    Here $F \in \mathbb{R}^n, F(x,t)$ denotes the temperature at node $x$ and time $t$. 

    \item \textcolor{black}{\textbf{Chemical Brusselator model}\;($\mathbb{C}$): 
    Introduced in 1971, the chemical Brusselator model exemplifies an autocatalytic chemical reaction system \cite{LANDSBERG1972, cencetti_clusella_fanelli_2018}. Its dynamics are governed by the equations:}
    \begin{equation}
    \left\lbrace \begin{array}{l}
\dot{x}_i \;=\;a-\left(b+d\right)x_i +c\;x_i^2 y_i -D_x \;\sum_j L {\;}_{\mathrm{ij}} x_j \\
{\dot{y}_i \;=\;{bx}}_i -c\;x_{\;i}^2 y_i -D_y \;\sum_j L_{\mathrm{ij}} \;y_j \;\;\;\;\;\;\;
\end{array}\right.
\label{Rd dynamics}
\end{equation}
     \item \textbf{Epidemic Spread\;$(\mathbb{E})$:}  Understanding how infectious diseases propagate through social networks can aid in designing effective strategies for disease control, information dissemination, and opinion formation. The SIS (susceptible-infected-susceptible) model, used to study epidemic spreading, is described as follows:
    $$
    \frac{dx_i}{dt} = -x_i + \sum_{j=1}^N A_{ij}(1-x_i)x_j.
    $$ $A_{ij}$ represents the $(i, j)$-th entry of the adjacency matrix of the graph. $N$ denotes the number of nodes.

    \item \textbf{Neural Dynamics}\;($\mathbb{N}$): The study of neural activity propagation in brain networks provides insights into cognition and neurological disorders. One such system, described in \cite{Hens2019}, follows the equation:
    $$
    \frac{dx_i}{dt} = -Bx_i + C \tanh{x_i} + \sum_{j=1}^N A_{ij}\tanh{x_j}.
    $$

\end{enumerate}

\subsection{Proper Orthogonal Decomposition for Dynamical systems} 
\label{POD-steps}
\textcolor{black}{This section provides a concise overview of the application of Proper Orthogonal Decomposition (POD) to initial value problems in dynamical systems \cite{rathinam_petzold_2003}.} Given a dataset $\mathcal{D}$ consisting of a collection of points $x^c_i$, where $x^c_i \in \mathbb{R}^n$ denotes the state of the system at time $t_i$ for a particular trajectory $c$, the POD method seeks a subspace $S \subset \mathbb{R}^n$ such that 
$$
\vert\vert \mathcal{D} - \rho_S \mathcal{D} \vert\vert^2  
$$
is minimized. $\rho_S$ is the orthogonal projection onto the subspace $S$ and $\rho_S \mathcal{D}$ denotes the projected data set.

$S \subset \mathbb{R}^n$ is the best $k$ dimensional approximating affine subspace, with the matrix $\rho$ of projection consisting of leading $k$ eigenvectors of the covariance matrix ($\bar{R}$). The subspace for a dataset $\mathcal{D}$ is uniquely determined by the projection matrix $P$ and the mean $\bar{x}$, $P = \rho^T \rho$.
\begin{align*}
\bar{R} &= \sum_{c=1}^{N_T} \int_{0}^T (x^c(t) - \bar{x})(x^c(t) - \bar{x})^T dt \\     
\bar{x} &= \frac{1}{N_{T} T} \sum_{c=1}^{N_T} \int_{0}^T x^c(t) dt.
\end{align*}

$N_T$ denotes the number of trajectories. 

An asymptotic and sensitivity analysis of POD is presented in \cite{rathinam_petzold_2003}. Consider a dynamical system in \( \mathbb{R}^n \) governed by a vector field \( f \):

\[
\dot{x} = f(x,t).
\]
\\
The reduced-order model (ROM) vector field is defined as:

\[
\dot{z} = \rho f(\rho^T z + \bar{x}, t) = f_{a}(z, t).
\]
Thus an initial value problem for the system $\dot{x} = f(x,t)$ with $x(0) = x_0$ using the projection method is given by, 
\begin{equation}
\label{equivalent_system}
\dot{\hat{x}} = Pf(\hat{x}, t); \;\; \hat{x}(0) = \hat{x}_0 = P(x_0 - \bar{x}) + \bar{x}.  
\end{equation}
\\
Here, \( \hat{x}_0 \) is the projection of \( x_0 \) onto the subspace \( S \).
 \textcolor{black}{The sensitivity of the POD projection \( P \) to changes in the dataset \( \mathcal{D} \) is quantified by the following proposition from \cite{rathinam_petzold_2003}:
}  
\\
\textbf{Proposition:} (Rathinam and Petzold \cite{rathinam_petzold_2003}) Consider applying POD to a data set $\mathcal{D}$ to find the best approximating $k(< n)$ dimensional subspace. Let the ordered eigenvalues of the covariance matrix of the data $\mathcal{D}$ be given by $\Tilde{\lambda}_1 \geq \cdots \geq \Tilde{\lambda}_n$. Suppose $\Tilde{\lambda}_k > \Tilde{\lambda}_{k+1},$ which ensures that $P(\mathcal{D})$ is well defined. Then 
\begin{equation}
S_k(\mathcal{D}) = \text{max}_{i\leq k, j\leq n-k} \,\sqrt{2}\; \frac{\sqrt{\Tilde{\lambda}_i + \Tilde{\lambda}_{j+k}}}{\Tilde{\lambda}_i - \Tilde{\lambda}_{j+k}} \sqrt{\Tilde{\lambda}_1 + \cdots + \Tilde{\lambda}_n} \geq \sqrt{2}. 
\label{sensitivity_pod}
\end{equation}

\subsection{Filtering}
\label{filteringsteppara}
\textcolor{black}{
This section introduces the filtering problem and outlines the key steps in the filtering framework. For further details on the concept of filtering, see \cite{lewis_lakshmivarahan_dhall_2009, resampling_matlab}.} 
\textcolor{black}{
Rooted in Bayesian inference, optimal filtering estimates the state of time-varying systems under noisy measurements. Its objective is to achieve statistically optimal estimation of the system state. Optimal filtering follows the Bayesian framework for state estimation, integrating statistical optimization with Bayesian reasoning to improve applications in signal processing, control systems, and sensor networks.}

\textcolor{black}{In Bayesian optimal filtering, the system state consists of dynamic variables such as position, velocity, orientation, and angular velocity. Due to measurement noise, observations do not yield deterministic values but rather a distribution of possible states, introducing uncertainty. The system state evolution is modeled as a dynamic system with process noise capturing inherent uncertainties in system dynamics. While the underlying system is often deterministic, stochasticity is introduced to represent model uncertainties.}

The system's state evolves according to:

\[
x_{k+1} = M(x_k) + w_{k+1},
\]

where \( x_k \in \mathbb{R}^n \) is the system state at time \( t_k \), and \( w_{k+1} \in \mathbb{R}^n \) represents the model error. The observations are given by:

\[
z_k = h(x_k) + v_k,
\]

where \( v_k \) represents the observation noise.

The filtering step estimates the state at time \( t_{k+1} \) (\( x_{k+1} \)) based on observations up to \( t_{k+1} \). 


\section{Problem statement}
\label{probstatement}
\textcolor{black}{Surrogate modeling techniques, such as the reduced-order model (ROM) discussed in Section \ref{POD-steps}, provide computational efficiency but are highly sensitive to input variations, as indicated in Equation \ref{sensitivity_pod} 
Consequently, when initialized with a new random condition, the surrogate model (\(M(\cdot)\)) efficiently approximates the system trajectory but may fail to capture deviations from the true state. \\
The surrogate model is used to obtain compressed representations of state vectors at various timesteps, which are treated as noisy observations. The observations are compressed for better memory efficiency if the model operates on the state dimension, as in the neural ODE surrogate model (Section \ref{benchmark}). Let $h: \mathbb{R}^n \rightarrow \mathbb{R}^k$ (where $k < n$) denote the state observation relationship. $e_{mh}(t_{k+1})$ and $e_{ml}(t_{k+1})$ denote the high and low wavelength components of error, capturing deviations between the true state and the state obtained from the surrogate model ($M(x_k)$) at time $t_{k+1}$. $e_{oh}(t_k)$ and $e_{ol}(t_k)$ denote the high and low-wavelength components of error that capture the difference between the observation and the state at time $t_k$.
 $x_{k+1}$ denotes the true state of the system at time $t_{k+1}$. $M(\cdot)$ denotes the forward model. For instance, if the forward model $M(\cdot)$ uses the Euler forward scheme and the surrogate model applies POD to the vector field $\dot{x} = f(x(t))$, then:
\[
M({x}_k) = {x}_k + h P f({x}_k, t_k),
\]
where $h$ denotes the step size. The state forward dynamics and state observation relationship are given by:} 
\begin{align*}
    x_{k+1} &= M({x}_k) + e_{mh}(t_{k+1}) + e_{ml}(t_{k+1}), \\
    z_k     &= h (x_k) + e_{oh}(t_k) + e_{ol}(t_k).
\end{align*}



\section{Proposed Framework}
\label{framework}

\begin{figure}
    \centering
    \includegraphics[width =0.9\linewidth, height = 8cm]{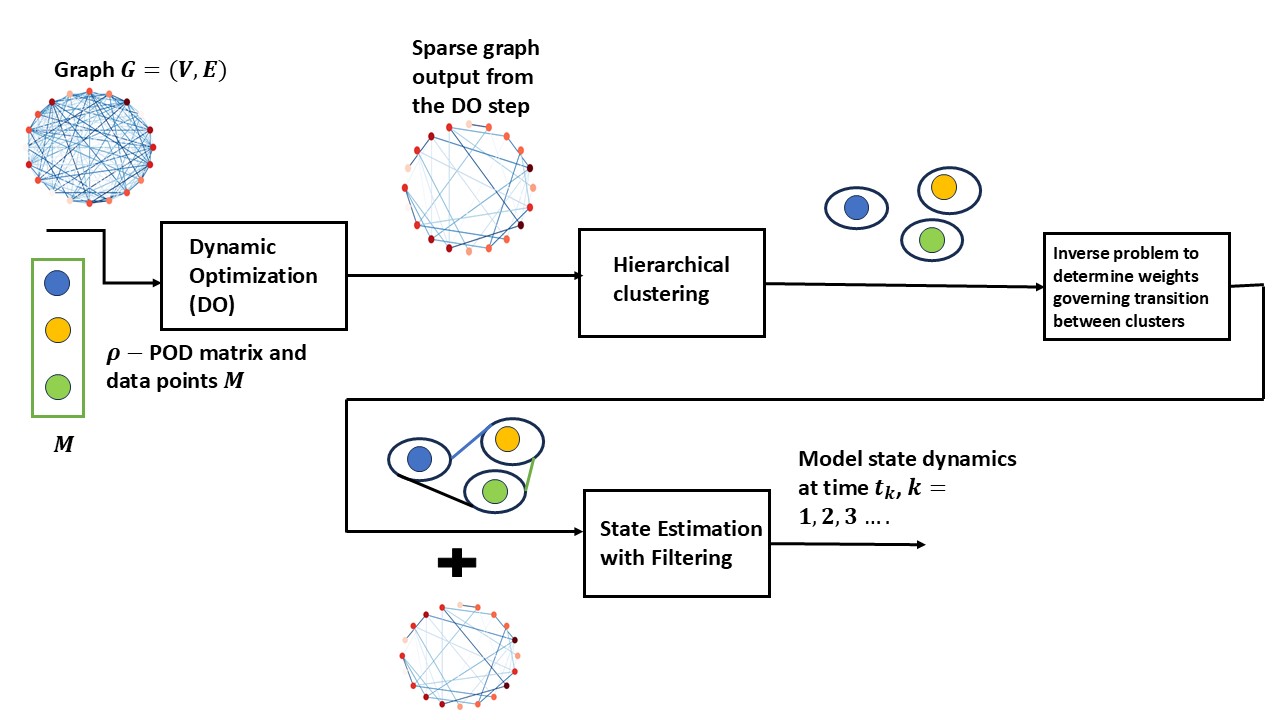}
    \caption{\textcolor{black}{Diagram illustrating the key steps of the proposed methodology, including dynamic optimization, hierarchical clustering, solving the inverse problem, and filtering. See Section~\ref{framework} for a detailed explanation.}}
    \label{Methodology}
\end{figure}

\textcolor{black}{The key steps of the framework, as illustrated in Figure \ref{Methodology}, are as follows. For brevity, the framework is discussed in the context of dynamical systems represented on graphs. For dynamical systems that do not involve graphs, the dynamic optimization step is not required. \textcolor{black}{The framework's applicability to general modeling scenarios is demonstrated in Section \ref{benchmark}, where a neural ODE-based surrogate model is used to learn from data.}}   
\begin{enumerate}
    \item \textbf{Dynamic optimization.}
    \textcolor{black}{For dynamical systems on graphs, a dynamic optimization problem is formulated to extract a sparse graph from known solution trajectory snapshots. The sparse graph improves memory efficiency in the filtering step (Step 4) of the framework. The original graph can be replaced with the sparse graph, further enhancing memory efficiency. This step outputs the graph Laplacian matrix of the sparse graph, given by \( L_1 = B^T \text{diag}(w^*) B \), where the weights \( w^{*} \) for linear and nonlinear dynamical systems on graphs are determined by solving the optimization problems detailed in Sections \ref{dynamiclinear} and \ref{dynamoptimrd}.
 }
 
    \item \textbf{Hierarchical clustering.} \textcolor{black}{This step quantifies deviations between the surrogate model (e.g., POD) and the actual system trajectories. Given a dataset \( \mathcal{D} = \{x_1, x_2, \ldots, x_T\} \) and surrogate model predictions \( \mathcal{P} = \{z_1, z_2, \ldots, z_T\} \), we apply hierarchical clustering to partition \( \mathcal{P} \) into \( p \) clusters. The assigned cluster labels are represented as \( Y = \{y_1, y_2, \ldots, y_T\} \), where \( y_t \) denotes the cluster assignment for data point \( z_t \). 
    For each cluster $C_i$, we construct a set of triplets $\{z_j, e_j, \hat{e}_j\}_{j=1}^{N_{C_i}}$, where $N_{C_i}$ is the number of triplets in cluster $C_i$. For reduced order model using POD (Section \ref{POD-steps}), we denote $e_j = x_j - (\rho^T {z}_j + \bar{x}), \hat{e}_j = {z}_j - \rho(x_j - \bar{x})$. The center of cluster $C_i$ is the triplet $\{{z}^{C_i}_{m}, e^{C_i}_m, \hat{e}^{C_i}_m\}$, where $e^{C_i}_m = \frac{\sum^{N_{C_i}}_{j=1} e_j}{N_{C_i}}, \hat{e}^{C_i}_m = \frac{\sum^{N_{C_i}}_{j=1} \hat{e}_j}{N_{C_i}}, {z}^{C_i}_m = \frac{\sum^{N_{C_i}}_{j=1} {z}_j}{N_{C_i}}$. As shown in Figure \ref{cluster}, a visualization of the clusters is provided, with $I_{C_i}$ representing the indices of data points within cluster $C_i$.}
    \item \textbf{Inverse Problem for Estimating Weights Between Clusters.} \textcolor{black}{For an initial condition $x_0$ at time step $t_0$, we obtain a compressed representation of the data point $z_0$. For example, when using the POD method, we get the compressed representation of $x_0$ as ${z}_0 = \rho (x_0 -\bar{x})$. From the clusters formed in Step 3, we assign the cluster to the point $z_0$ based on the following, 
    \begin{equation}
\arg\min_{f \in \{1, 2, \dots, p\}} \left\| z^{C_f}_m - {z}_0 \right\|_2
.   \label{equationinversefirst}     
    \end{equation}
    The initial cluster distribution \( p_0 \in \mathbb{R}^p \) is a one-hot vector, where the index corresponding to \( f \in \{0,1,\ldots,p\} \) is set to 1. The cluster distribution at time \( t \) is denoted by \( p_t \). For successive transitions \( \hat{x}_1 \rightarrow \hat{x}_2 \rightarrow \ldots \rightarrow \hat{x}_T \), it is essential to efficiently determine the cluster of each approximate state \( \hat{x}_i \) generated by the surrogate model at time \( t_i \). This step is performed for the estimation of high wavelength components of errors described in Section \ref{probstatement}.
 The transitioning of cluster distribution ($p_{t}$) is modelled as a Markovian process discussed in \cite{Chung1997}. We consider the nodes of the clusters from Step 2 as nodes of graph $\mathcal{G}_{I} = (\mathcal{V}_{I}, \mathcal{E}_{I}, q)$, $\vert \mathcal{V}_{I} \vert = p$, where the weights of the graph $q$ are unknown. The topology of the graph is considered complete with self-loops, meaning there is an edge between each node of the cluster and an edge from each node to itself. We use a supervised learning approach to determine the weights $q$ within the graph, where the labels $y_{t}$ from Step 2 are used in the optimization objective function with the cross-entropy loss (See \ref{P}).
}

\par \textcolor{black}{The transitions of the cluster distribution is modeled as Markovian, based on the assumption that the states of the ODE exhibit Markovian properties, which is explained below.} The general explicit Runge-Kutta numerical scheme using $n$ slopes (\cite{burden_faires_burden_2016}) for obtaining solution of the differential equation $\dot{x} = f(x(t))$ is given by the following relation $$x_{t+1} = x_t + h b_1 k_1 + h b_2 k_2 + \ldots h b_n k_n,$$
        $$
        k_i = hf(x_i + h \sum_{j=1}^{i-1} a_{ij} k_j).
        $$
From this formulation, it follows that the solution transitions of the ODE adhere to a discrete-time Markov chain:
\[
p(x_{t+1} \mid x_0, x_1, \dots, x_t) = p(x_{t+1} \mid x_t).
\] 
        In a graph $\mathcal{G}_{I}$, a walk is represented as a sequence of vertices $(v_0, v_1, \dots, v_s)$, where each consecutive pair of vertices $(v_{i-1}, v_{i})$ is connected by an edge in $\mathcal{G}_{I}$. In other words, there is an edge between $v_{i-1}$ and $v_{i}$ for all $1 \leq i \leq s$. A random walk is characterized by the transition probabilities \( P(u, v) \), which define the probability of moving from vertex \( u \) to vertex \( v \) in one step. Each vertex \( u \) in the graph can transition to multiple neighboring vertices \( v \). If we consider the transitions between the clusters as Markovian, the transition matrix $P = D^{-1} A$, where $D$ is the diagonal matrix of degrees and $A$ is the adjacency matrix of the graph. For each vertex $u$,
$$\sum_v P(u,v) = 1.$$
\begin{equation*}
P\left(u,v\right)=\;\left\lbrace \begin{array}{ll}
\frac{1}{d_u} & \; \text{if u and v are adjacent},\\
0 & \mathrm{otherwise}
\end{array}\right.    
\end{equation*} If we consider a complete graph on 3 nodes with self-loops as shown in Figure \ref{cluster}, the transition matrix, 
        $$P = \left\lbrack \begin{array}{ccc}
\frac{1}{q_{11} + q_{12} + q_{13}\;} & 0 & 0\\
0 & \frac{1}{q_{22} + q_{23} + q_{12}} & 0\\
0 & 0 & \frac{1}{q_{13} + q_{23} + q_{33}}
\end{array}\right\rbrack \left\lbrack \begin{array}{ccc}
q_{11} & q_{12} & q_{13}\\
q_{12} & q_{22} & q_{23}\\
q_{13} & q_{23} & q_{33}
\end{array}\right\rbrack. $$

\begin{figure}[H]
\centering
    \includegraphics[width=0.7\textwidth]{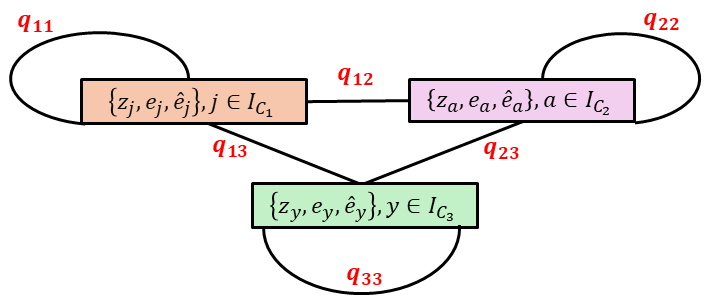} 
\caption{Illustration of triplet organization into clusters, with mutually exclusive sets of indices \( I_{C_i} \) (where \( i = 1, 2, 3 \)). The inverse problem \ref{P} is used to estimate weights \( q_{ij} \).
}
\label{cluster}
\end{figure}
We now pose a constrained optimization problem~\ref{P} to estimate the weights ($q$) responsible for governing the transitions of data points.
\begin{equation*}
\left.
\begin{aligned}
& \textbf{minimize\;\;}_{q \, \in \, \mathbb{R}^{n(n+1)/2}} \;\;&&J = \sum_{t=1}^T \sum_{k=1}^{p} -({y}_t) \log ({p}_t(k))   - (1-{y}_t) \log(1 - ({p}_t(k)))\\
& \textbf{subject\;\; to \;\;} &&{p}_t  = P(q) \; {p}_{t-1} \\
& && q_{ij}  \geq 0 \;\;i,j = 1,2,\ldots p.\\
\end{aligned}
\right\} \tag{P} \label{P}
\end{equation*}
\textcolor{black}{${y}_t$ represents the cluster index of data point $z_t$}. \textcolor{black}{The optimization problem above is solved to find the optimal \( q^{*} \), which can be efficiently obtained using the adjoint method for data assimilation, as detailed in \cite{lewis_lakshmivarahan_dhall_2009}. Given the cluster distribution ($p_k$) at time $t_k$, we estimate components $A_k, R_k$ (Eq. \ref{modelstaterelation}) for the high wavelength component of errors $e_{mh} (t_k), e_{oh}(t_k)$ as follows:
$A_k = \text{diag(}{\sum_{i=1}^p p_k(i) {e}_m^{C_i}}), R_k = \text{diag(}{\sum_{i=1}^p p_k(i) \hat{e}_m^{C_i}}).$
}  \\
\item \textbf{State Estimation with Filtering.} \textcolor{black}{
The goal of filtering is to estimate system states over time using noisy measurements and state observation relations, as discussed in Section \ref{filteringsteppara}. The POD surrogate model defines the evolution of state dynamics and the state-measurement relationship as follows:
\begin{equation}
\label{modelstaterelation}
\begin{cases}
x_{k+1} = M(Pf({x}_k)) + A_{k+1} v_{k+1} + w_{k+1}, &  \\
z_k = \rho (x_k - \bar{x}) + R_k \beta_k + \mu_k, & 
\end{cases}
\end{equation} 
The terms \( A_{k+1} v_{k+1} \) and \( w_{k+1} \) represent the high- and low-wavelength components of model error at time \( t_{k+1} \), respectively. Similarly, \( R_k \beta_k \) and \( \mu_k \) characterize the high- and low-wavelength components of the state-observation error at time \( t_k \). $P = \rho^T \rho$ denotes the projection matrix. 
At a particular time step, the data point at time $t$ $({z}_t)$ is obtained using the ROM $\dot{{z}}_t =  f_a ({z}_t, t)= \rho f(\rho^T {z}_t + \bar{x}, t)$, and the state vector after projection ${x}_t = \rho^T {z}_t + \bar{x}$. 
The matrices \( A_{k+1} \) and \( R_k \) are computed following the methodology outlined in Step 3 of the framework. The term $v_1$ is the solution to the linear system $A_1  v_1 = x_1 - M(Pf(x_0))$. For dynamical systems represented on graphs, the term $x_1$ is computed using an explicit numerical scheme with the sparse graph obtained from the dynamic optimization Step 1. 
In the case of a diffusion equation represented on graphs $(\mathbb{D}$ in Section \ref{spatio-temporal})  with Euler discretization step size $h$, the sparse Laplacian matrix $L_1$ from the output of Step 1 is used instead of matrix $L$ to obtain $x_1 = x_0 - h L_1 x_0$. $\beta_1$ is the solution to the linear system $R_1 \beta_1 = z_1 - \rho(x_1 - \bar{x})$. 
$w_{k+1}$ and $\mu_k$ are modelled as Gaussian with means $\mathbf{0}_{n\times1}, \mathbf{0}_{k \times 1}$ and variances $\zeta_x = \sigma_x I_{n\times n}, \zeta_y = \sigma_y I_{k \times k}.$ The distribution governing the low wavelength components of errors ($w_{k+1}, \mu_k$) is deliberately shaped to exhibit minimal variability, with a centered mean of $\mathbf{0}_{n\times 1}$. This assumption is made with the expectation that the high wavelength component will accurately capture the true nature of the error. 
There exist linear and non-linear filters for state estimation as described in \cite{lewis_lakshmivarahan_dhall_2009}, \cite{resampling_matlab}. When the uncertainty in the state prediction exceeds a threshold, vectors $v_{k+1}$ and $\beta_k$ are updated as solutions of the linear systems defined below:
\begin{align}
& A_{k+1} \; v_{k+1} = x_{k+1} - M(Pf(x_{k})),  \\
& R_{k} \; \beta_k = z_{k} - \rho(x_{k} - \bar{x}).  
\end{align}
The computation of $x_{k+1}$ is performed using an explicit numerical scheme that leverages the state vector obtained from the filter at time step $k$ using the sparse graph.
}
     
\end{enumerate}
  
\section{Dynamic optimization problem formulation for dynamical systems represented on graphs}
 \label{linearsystem}
\textcolor{black}{This section outlines the constraints necessary for the dynamic optimization step (Step 1 in Section \ref{framework}) applied to a linear dynamical system \( \mathbb{D} \) represented on graphs (Section \ref{spatio-temporal}). Section \ref{dynamiclinear} details the dynamic optimization process for linear diffusion systems on graphs. Section \ref{dynamoptimrd} presents the optimization framework for reaction-diffusion systems on graphs, including the corresponding constraint formulations.} \\
 Existing techniques for approximating linear dynamical systems on graphs include \cite{lai2021graph}, \cite{spielman_teng_2011}, and \cite{spielman2009graph}. The approach in \cite{spielman2009graph} constructs spectral sparsifiers by sampling edges according to their resistance \( R_e \). 
 The probability of sampling an edge \( e \) to construct a sparse approximation \( \bar{G} \) of graph \( G \) is given by  
\[
p_e = \frac{w_e R_e}{(n-1)},
\]
where the edge resistance is defined as  
\[
R_e = \| L^{1/2} L_e L^{1/2} \|.
\]
See Algorithm~\ref{Sparsify} for further details.
\[
\begin{tabular}{|l|}
\hline
$\bar{G} = \textbf{Sparsify}(G,q) \cite{spielman2009graph}$ \\
\hline
Choose a random edge $e$ of $G$ with probability $p_e$ proportional to $w_e R_e$, \\
and add $e$ to $\bar{G}$ with weight $w_e/(q p_e)$. Take $q$ samples independently \\
with replacement, summing weights if an edge is chosen more than once. \\
\hline
\end{tabular}
\tag{A}
\label{Sparsify}
\]
\textcolor{black}{Computing bounds on the edge weights of the sparse graph approximation \( \bar{G} \), using effective resistances or local edge connectivities, is computationally intensive. To overcome this challenge, we propose Algorithm~\ref{algo_degree}, which leverages Theorems~\ref{weight_upper} and~\ref{connec_theorem} to efficiently compute upper bounds on edge weights and node degrees in the sparse graph. The following discussion introduces Bernstein's Inequality \cite{boucheron2013concentration}, which is utilized in the derivation of Theorem~\ref{weight_upper}.}
\begin{theorem}
\textbf{Bernstein's Inequality:} Suppose $X_1, \ldots, X_n$ are independent random variables with finite variances, and suppose that $\text{max}_{1 \leq i \leq n} \vert X_i \vert \leq B$ almost surely for some constant $B > 0$. Let $V = \sum_{i=1}^n \mathbb{E}X_i^2.$ Then, for every $t \geq 0,$ 
$$ P\{\sum_{i=1}^n (X_i - \mathbb{E}X_i) \geq t \} \leq exp\bigg(-\frac{t^2}{2(V + tB/3)}\bigg)$$ \\
and \\
$$ P\{\sum_{i=1}^n (X_i - \mathbb{E}X_i) \leq -t \} \leq exp\bigg(-\frac{t^2}{2(V + tB/3)}\bigg)$$
\label{Bernsteins}
\end{theorem}

\begin{theorem}
Let each edge \( (u, v) \) of a graph be sampled with probability \( p_{uv} \propto R_{uv} \). where $p_{uv} \geq \frac{\beta}{n \text{min(deg(u), deg(v))}}$, then $\mathbb{P}\{\frac{\sum_{i=1}^q X_i}{q} - w_{uv}  \geq t \} \leq \frac{1}{mn}$, $t \geq \frac{\epsilon_1 c_1}{3q} + \sqrt{\frac{2\epsilon_1 w_{uv} c_1}{q} + (\frac{\epsilon_1 c_1}{3q})^2}$, where 
$$X_i = \left\lbrace \begin{array}{ll}
\begin{array}{l}
\frac{w_{uv}}{p_{uv}},   \text{\;\;with probability $p_{uv}$ (u,v) $\in \; E(G)$} \\ \\
0, \text{\;\;otherwise.} \\ \\
\end{array} & 
\end{array}\right.$$
$\epsilon_1 = \log(nm), \vert X_i \vert \leq c_1, c_1 = \frac{w_{uv} n c_2}{\beta}$, where $c_2 = \text{min(deg($u$), deg($v$))}$.
\label{weight_upper}
\end{theorem}
\begin{proof}
The expectation of \( X_i \) satisfies  
\[
\mathbb{E}[X_i] = w_{uv}.
\]
The second moment sum satisfies  
\[
\sum_{i=1}^q \mathbb{E} (X_i^2) = q \frac{w_{uv}^2}{p_{uv}} \leq \frac{w_{uv}^2 n c_2 q}{\beta},
\]
where \( c_2 = \min(\text{deg}(u), \text{deg}(v)) \).

    Applying Bernstein’s Inequality (Theorem~\ref{Bernsteins}), we obtain  
\[
\mathbb{P}\left\{\frac{\sum_{i=1}^q X_i}{q} - w_{uv}  \geq t \right\} \leq \exp\left(-\frac{(tq)^2}{2} \cdot \frac{1}{w_{uv} c_1 q + \frac{c_1 t q}{3}}\right).
\]

    $$
 \frac{(tq)^2}{2} \frac{1}{(w_{uv} c_1 q + \frac{c_1 t q}{3})} \geq \log{nm} \equiv \epsilon_1   
    $$
    Completing the square yields  
\[
\left( t - \frac{\epsilon_1 c_1}{3q} \right)^2 \geq \frac{2\epsilon_1 w_{uv} c_1}{q} + \left(\frac{\epsilon_1 c_1}{3q}\right)^2.
\]
Hence, it suffices to choose  
\[
t \geq \frac{\epsilon_1 c_1}{3q} + \sqrt{\frac{2\epsilon_1 w_{uv} c_1}{q} + \left(\frac{\epsilon_1 c_1}{3q}\right)^2}.
\]

\end{proof} 
\textcolor{black}{Theorem \ref{weight_upper} provides upper bounds on the edge weights of the sparse graph \( \bar{G} \) generated by Algorithm \ref{Sparsify}. Furthermore, it is necessary to establish bounds on the node degrees in the sparse graph \( \bar{G} \). To derive these bounds, we introduce Theorem~\ref{connec_theorem}, which utilizes the established relationship between the eigenvalues of the Laplacian matrix and the node degrees, as discussed in Section~\ref{knownbounds}.}\\
\textbf{Notation:}  
Let \( G \) be a simple weighted graph, i.e., a graph without self-loops or multiple edges with vertices $\{v_1, v_2, \ldots, v_n\}$. We define two subsets of vertices in \( G \) based on their degrees. The set \( l_m \) contains the \( n-m+1 \) vertices with the highest degrees, while \( s_m \) consists of the \( m \) vertices with the lowest degrees. The subgraphs \( G_{l_m} \) and \( G_{s_m} \) are defined as the subgraphs of \( G \) induced by the vertex sets \( l_m \) and \( s_m \), respectively.


\begin{equation}
\begin{aligned}
    a_{ub} &= \quad \text{upper bound on the maximal edge weight of the sparsified graph}, \\
    \Delta (G_{l_i})_{ub} &= \quad \text{upper bound on the maximum degree in } G_{l_i}, \\
    \delta (G_{s_i})_{lb} &= \quad \text{lower bound on the minimum degree in } G_{s_i}.
\end{aligned}
\end{equation}
Here $d_{v_1} = d_1 \leq d_{v_2} = d_2 \leq \ldots \leq d_{v_n} = d_n.$ 
\subsection{Some known bounds on eigen values of Laplacian matrix }
\label{knownbounds}
\underline{\textbf{Unweighted graphs}}: \cite{farber2011upper}(corollary 1) For simple unweighted graphs on $n$ vertices and $G \neq K_{n - m +1} + (m-1)K_1$ and $\bar{G} \neq K_{m-1} + (n-m+1)K_1$, we have the following relation
$$ d_m -n + m + 1 \leq \lambda_m (G) \leq d_{m-1}(G) + m - 2.$$
\underline{\textbf{Weighted graphs}}: \cite{farber2011upper}(corollary 2) Let $G$ be finite simple weighted graph on $n$ vertices and denote by a the maximal weight of an edge in $G$, then 
\begin{equation}
d_m(G) - \Delta(G_{l_m}) \leq \lambda_m(G) \leq d_{m-1}(G) + (m-1)a - \delta(G_{S_{m-1}}) \label{bounds}    
\end{equation}
\textcolor{black}{We present the following theorem, which establishes bounds on the degrees of the spectral sparsifier \( \bar{G} \) of \( G \) which are incorporated as constraints in the dynamic optimization step of our framework (Section~\ref{dynamiclinear}).}
\begin{theorem}
Let $\bar{G}$ be an $\epsilon$-spectral sparsifier of an undirected graph $G$ with $n$ vertices, where the eigenvalues of $G$ satisfy $\lambda_{1} \leq \lambda_{2} \leq \cdots \leq \lambda_{n}$, with $\lambda_n$ as the largest eigenvalue. The following bounds on the degrees of $\bar{G}$ must hold:
\newline

\begin{equation*}
d_{2}(G) - \Delta(G_{l_2})(1 - \epsilon) - a(G) \leq d_1(\bar{G}) \leq d_1(G)(1+\epsilon),
\end{equation*}

\begin{equation*}
(d_{i+1}(G) - \Delta(G_{l_{i+1}}))(1-\epsilon) - i a(\bar{G}) + \delta(\bar{G}_{s_i})_{lb} \leq d_i(\bar{G}) 
\leq (1+\epsilon) (d_{i-1}(G) + (i-1)a(G) - \delta(G_{s_{i-1}})) + \Delta(\bar{G}_{l_i})_{ub}, 
\end{equation*}
\hspace{75mm} \text{for } i = 2 \text{ to } n-1.

\begin{equation*}
d_n(G)(1 - \epsilon) \leq d_n(\bar{G}) \leq (1 + \epsilon)(d_{n-1}(G) + (n-1)a(G) - \delta(G_{s_{n-1}})).
\end{equation*}
\label{connec_theorem}
\end{theorem}
\begin{proof}
Since $\bar{G}$ is an $\epsilon$-spectral approximation of $G$, the eigenvalues satisfy:
\begin{equation}
(1-\epsilon)\lambda_i(G) \leq \lambda_i(\bar{G}) \leq (1+\epsilon)\lambda_i(G), \tag{a}
\end{equation}
for $i = 1, \dots, n$.

From the spectral bounds on eigenvalues in Equation~\ref{bounds}, we know:
\begin{equation}
\lambda_i(G) \leq d_{i-1}(G) + (i-1)a(G) - \delta(G_{s_{i-1}}), \quad 
\lambda_i(G) \geq d_i(G) - \Delta(G_{l_i}). \tag{b}
\end{equation}

Applying (a) to the upper bound in (b), we obtain:
\begin{equation}
d_i(\bar{G}) \leq (1+\epsilon)\big(d_{i-1}(G) + (i-1)a(G) - \delta(G_{s_{i-1}})\big) + \Delta(\bar{G}_{l_i})_{ub} . \tag{c}
\end{equation}

Similarly, using the lower bound in (b), we have:
\begin{equation}
d_{i-1}(\bar{G}) \geq \big(d_i(G) - \Delta(G_{l_i})\big)(1-\epsilon) - (i-1)a(\bar{G}) + \delta(\bar{G}_{s_{i-1}})_{lb}. \tag{d}
\end{equation}

For the largest degree, we use the fact that $\bar{G}$ preserves cuts in $G$, giving:
\begin{equation}
d_n(G)(1-\epsilon) \leq d_n(\bar{G}) \leq (1+\epsilon)(d_{n-1}(G) + (n-1)a(G) - \delta(G_{s_{n-1}})). \tag{e}
\end{equation}

Hence, the degree bounds for all nodes in $\bar{G}$ are established.
\end{proof}
Let $\delta_i^u$ and $\delta_i^l$ denote the upper and lower bounds on the degree of node $i$, respectively, as established in Theorem \ref{connec_theorem}. Additionally, let $\delta^{-}(i)$ represent the refined lower bound of degree $i$, and $\Delta^{+}(i)$ the refined upper bound:
\[
d_i(\bar{G}) \geq \max\left(\delta_i^l, (1-\epsilon)\,d_i(G)\right) = \delta^{-}(i),
\]

\[
d_i(\bar{G}) \leq \min\left(\delta_i^u, (1+\epsilon)\,d_i(G)\right) = \Delta^{+}(i).
\]

Algorithm \ref{algo_degree} outlines the computation of the upper bound $\Delta^+$ and the lower bound $\delta^-$ on the node degrees in the sparsified graph $\bar{G}$.  

Here, $\Delta^+(i)$ denotes the upper bound on the maximum degree in the subgraph of $\bar{G}$ induced by the vertices $\{v_i, v_{i+1}, \dots, v_n\}$, while $\delta^-(i)$ represents the lower bound on the minimum degree in the subgraph induced by $\{v_1, v_2, \dots, v_{n-i+1}\}$. Furthermore, we define upper and lower bounds on graph $G$ as $\Delta(i) = \Delta(G_{l_i})$ and $\delta(i) = \delta(G_{s_i})$.

\begin{algorithm}[H]

---------------------------------------------------------------------------------------------------------------------\\
\textbf{Algorithm 1 \\}
---------------------------------------------------------------------------------------------------------------------\\
\SetKwInOut{Input}{Input}
\SetKwInOut{Output}{Output}

\Input{Graph $G = (V,E, W)$, $\epsilon$}
\Output{$\Delta^+, \delta^-, \Delta, \delta, w^+, w^-$}
$w^+ = \{ \}; \;\; w^- = \{\}$ \\
\For{each edge \( e \in E(G) \)}{
    Compute \( t \) using Theorem \ref{weight_upper}. \\
    Append \( w_e + t \) to \( w^+ \). \\
    Append \( \max(w_e - t, 0) \) to \( w^- \).
}
$\text{node} \leftarrow \{1,2,\ldots, n\}$ \\
degree $\leftarrow\; \{d_1, d_2, \ldots, d_n\}$ dfs($G$)   // O($V+E$) \\
$\text{degree}, \text{degree1}, \text{degree}^+, \text{degree}^- \leftarrow \text{dictionary($\{v_1, v_2, \ldots, v_n\}$, $\{d_{v_1}, d_{v_2},\ldots, d_{v_n}\}$) (where $d_{v_1} \leq d_{v_2} \leq \ldots \leq d_{v_n}$)}$

$\Delta(1) = \text{max(degree)}, \delta(n) = \text{min(degree)}, a(G) (\text{maximal weight})$ \\
$\Delta(n) = 0, \;\;\Delta^+(n) = 0$\\
$\delta(1) = 0, \;\;\delta^-(1) = 0$

$\Delta^+(1) = \text{max(degree)}, \delta^-(n) = \text{min(degree)}$ \\
$l_2$ $\leftarrow$ node; $l_3$ $\leftarrow$ node \\
\text{Remove node $v_{1}$ from $l_2$} \newline 
\text{Remove node $v_{n}$ from $l_3$} \newline 
\For{$i = 2$ to $n-1$}{
    
    \For{$y = \{v_{i+1}, v_{i+2}, \ldots, v_n\}$}
    {
    if $v_i, y \in E$ \\ 
        \hspace{7mm}degree($y$) = degree($y$) - $w(v_i, y)$, degree($v_i$) = degree($v_i$) - $w(v_i, y)$ \\ \hspace{7mm}$\text{degree}^+(y) = \text{degree}^+(y) - w^+(v_i, y)$, $\text{degree}^+(v_i) = \text{degree}^+(v_i) - w^+(v_i, y)$
    
    } 
   
    \text{Remove node $v_{i}$ from $l_2$} \newline 
    \For{$y = \{v_{1}, v_{2}, \ldots, v_{n-i}\}$}
    {
    if $v_{n-i+1}, y \in E$ \\ 
        \hspace{7mm}degree1($y$) = degree1($y$) - $w(v_i, y)$, degree1($v_i$) = degree1($v_i$) - $w(v_i, y)$ \\ \hspace{7mm}$\text{degree}^-(y) = \text{degree}^-(y) - w^-(v_i, y)$, $\text{degree}^-(v_i) = \text{degree}^-(v_i) - w^-(v_i, y)$
    
    }
    
    $\delta(n-i+1)$ = \text{$\min_{x \in l_3} \text{degree1}$} \\
    $\delta^-(n-i+1)$ =  $\min_{x \in l_3} \text{degree}^-$ 
    \\
    \text{Remove node $v_{n-i+1}$ from $l_3$} \\
}
\label{algo_degree}
\end{algorithm}
\textcolor{black}{
The outputs of Algorithm~\ref{algo_degree} (\( \Delta^+, \delta^-, w^+, w^- \)) are incorporated as constraints in the dynamic optimization process (Constraints \ref{D3}, \ref{D4}, \ref{D5}, \ref{D6}). These constraints are specifically applied in the dynamic optimization of linear dynamical systems on graphs, as discussed in Section~\ref{dynamiclinear} and outlined in Section~\ref{framework}.
}




\subsection{\textcolor{black}{Dynamic Optimization for Linear Dynamical Systems on Graphs Using a POD-Based Surrogate Model}} 
    \label{dynamiclinear}
    Real-time dynamic optimization problems are described in \cite{book}, \cite{bartusiak}, and \cite{Nagy2007}). Using snapshots of solutions from several trajectories (Section \ref{spatio-temporal}), we formulate the dynamic optimization problem in a ROM space \cite{rathinam_petzold_2003}, \cite{ajayakumar2023data} using the matrix $\rho \in \mathbb{R}^{k \times \vert V\vert }$ of projection and mean $\bar{x}$, where $k$ denotes the reduced dimension and $k < n$ (Section \ref{POD-steps}). The cost function $J$ (Eq. \ref{D1}) in the dynamic optimization makes use of data points obtained from ROM. The dynamic optimization framework for obtaining a sparse graph in the case of diffusion ($\mathbb{D}$) considering a single trajectory is given below: 
    \begin{equation}
    \label{D1}
\underset{\bar{\gamma}}{\textbf{minimize}} \,\, J(\bar{z},\bar{\gamma}) = 
\underbrace{\frac{\sum_{i=1}^{N}
      \sum_{j = 1}^k (z(i,j)(i \;\text{step } m_1)  -\bar{z}(i,j)(i \;\text{step } m_1) \;)^2 }{2}}_{A1}  
      +  \underbrace{{\alpha} \sum_{i=1}^m \vert \bar{\gamma}_i \vert}_{A2}
\end{equation}

\begin{alignat*}{2}
\textbf{subject\,\, to} \\
& \tag{D2} \label{D2} t_n \;N f_a(\bar{z}(i,1), \bar{z}(i,2), \ldots, \bar{z}(i,k), \bar{\gamma}) = \left\lbrack \begin{array}{c}
\bar{z}(i, 1)\\
\bar{z}(i, 2)\\
\ldotp \\
\ldotp \\
\ldotp \\
\bar{z}(i, k)
\end{array}\right\rbrack  - \left\lbrack \begin{array}{c}
\bar{z}(i-1, 1)\\
\bar{z}(i-1, 2)\\
\ldotp \\
\ldotp \\
\ldotp \\
\bar{z}(i-1, k)\\
\end{array}\right\rbrack && \;\;i = 1,2,\ldots C \\ &\{\bar{z}(0,1),\bar{z}(0,2),\ldots, \bar{z}(0,k)\} = \rho (F_0 - \bar{F}) \\
& \tag{D3} \label{D3}Q^{'} \bar{\gamma} \geq \delta^{-} \quad &&\\
&\tag{D4} \label{D4}Q^{'} \bar{\gamma} \leq \Delta^{+} \quad &&\\
& \tag{D5} \label{D5}w_j \bar{\gamma}_j \geq w^{-}_j, &&j = 1,2,\dots,m \\
& \tag{D6} \label{D6}w_j \bar{\gamma}_j \leq w^{+}_j, &&j = 1,2,\dots,m \\
& \tag{D7} \label{D7}\bar{\gamma}_j \geq 0, &&j = 1,2,\dots,m \\
\end{alignat*}
    
\setlist[enumerate]{resume}
The term $A1$ in the objective function (Eq. \ref{D1}) aims to minimize the discrepancy between the states obtained using the projected vector field with multipliers $\bar{\gamma}$ and the known ROM solution obtained from the original graph. The term $A2$ enforces sparsity in $\bar{\gamma}$ by penalizing nonzero elements. The optimization considers states at discrete increments of step size $m_1$. The matrix $N$ is determined by the chosen collocation method (see \cite{HEDENGREN2014133}).

$t_n$ denotes the step size. $C$ denotes the number of collocation elements. $z(i,k)$ represents the $k-$th entry of state $z$ in the $i$-th collocation element. $F_0 = \{F(0,1), F(0,2), \ldots, F({0, \vert V \vert})\}$ represents the initial condition used for generating data using ROM $z(i,j)$. $\bar{F}$ is determined based on the procedure described as in \cite{rathinam_petzold_2003}. The first set of $kC$ constraints stems from the orthogonal collocation method on finite elements (Eq. \ref{D2}). Non-negativity of the multipliers $\bar{\gamma}_j$ is imposed in constraints (Eq. \ref{D7}). The constraints $Q'\bar{\gamma} \geq \delta^-$ (Eq. \ref{D3}) and $Q'\bar{\gamma} \leq \Delta^+$ (Eq. \ref{D4}) enforce connectivity levels as derived in Theorem \ref{connec_theorem}. Here, $Q^{'} = Q \;\mathrm{diag(}w)$, where $Q$ represents the unsigned incidence matrix of the graph and $w$ represents the weights of the given graph. The weight constraints $w_j \gamma_j \geq w^{-}_j$ and $w_j \gamma_j \leq w^{+}_j$ are based on Theorem \ref{weight_upper}. From the output of the dynamic optimization framework $(z^*, \bar{\gamma}^*)$, 
we prune out weights $w^*_j \leq \epsilon_1$, where $\epsilon_1$ is a user-defined parameter ($w^{*}_j = w_j \gamma^{*}_j$). The new weight vector is denoted as $w_1^*$. The sparse graph Laplacian $L_1 = B^T \text{diag}(w_1^*) B$ replaces the standard Laplacian $L$ in the filtering step to improve the efficiency of Laplacian-vector product computations. 
\par Several key observations can be made about the dynamic optimization problem. The objective function ($J$) can be made continuous by using the substitution $\gamma = \gamma^+ - \gamma^-$, $\vert\vert \gamma \vert\vert_1 = \mathbf{1}^T \gamma^+  + \mathbf{1}^T \gamma^-$. $\gamma^+$ represents the positive entries in $\gamma$ (0 elsewhere) and $\gamma^-$ represents the negative entries in $\gamma$ (0 elsewhere). The objective function is convex, and the constraints are continuous. To solve this optimization problem, we can employ the Barrier method for constrained optimization, as discussed in \cite{nocedalbook} and \cite{luenberger2021}.

\subsection{\textcolor{black}{
Dynamic Optimization for Reaction-diffusion Dynamical Systems on Graphs Using a POD-Based Surrogate Model}}
\label{dynamoptimrd}
In this section, we present the formulation of the dynamic optimization framework for a reaction-diffusion system. A general reaction-diffusion system is described in \cite{cencetti_clusella_fanelli_2018}, where the activity at node $j$ at time $t$ is represented by an $\mathrm{m}$-dimensional variable $r_j(t)$. The temporal evolution of $r_j$ follows the differential equation:
$$ \frac{dr_j}{dt} = \mathcal{F}(r_j) + K \sum_{k=1}^{N} A_{jk} \mathcal{G}(r_k-r_j) \hspace{5mm} j = 1,2,....N.$$
$\mathcal{F}$ denotes the reaction component, while the remaining terms represent the diffusion process. $\mathcal{F} : \mathbb{R}^{\mathrm{m}} \rightarrow \mathbb{R}^{\mathrm{m}}$, $\mathcal{G}: \mathbb{R}^{\mathrm{m}} \rightarrow \mathbb{R}^{\mathrm{m}}$. For brevity, we use the alternating self-dynamics Brusselator model ($\mathbb{C}$) discussed in Section \ref{Background}.
When the input graph is connected, We impose a minimum connectivity constraint to limit perturbations in the second smallest eigenvalue of the graph Laplacian matrix. The intuition behind this constraint is that for a connected graph, the lowest degree will be greater than zero, making this constraint necessary. $\tau_L$ is the minimum degree imposed on the sparse graph (Eq. \ref{D10}). Non-negativity of weights is also imposed on the dynamic optimization problem (Eq. \ref{D11}). While \cite{ajayakumar2023data} discusses generating sparse graphs using snapshots of data at arbitrary time points, we utilize data points at collocation time points for sparsification. We consider the following dynamic optimization problem for a reaction-diffusion system considering a single trajectory:

\[
\tag{D8} 
\underset{\bar{\gamma}}{\textbf{minimize}} \,\, J(\bar{z},\bar{\gamma}) = {\frac{\sum_{i=1}^{N}
      \sum_{j = 1}^k (z(i,j)(i \;\text{step $m_1$})  -\bar{z}(i,j)(i \;\text{step $m_1$}) \;)^2 }{2}}  +  {{\alpha} \sum_{i=1}^m \vert \bar{\gamma}_i \vert}
      \label{costfunc}
\]
\label{costfunc2}
\vspace{-6mm}
\hspace{-20mm}
\begin{alignat*}{2}
\textbf{subject\,\, to} \\
& \tag{D9} \label{D9} t_n \;N f_a(\bar{z}(i,1), \bar{z}(i,2), \ldots, \bar{z}(i,k), \bar{\gamma}) = \left\lbrack \begin{array}{c}
\bar{z}(i, 1)\\
\bar{z}(i, 2)\\
\ldotp \\
\ldotp \\
\ldotp \\
\bar{z}(i, k)
\end{array}\right\rbrack  - \left\lbrack \begin{array}{c}
\bar{z}(i-1, 1)\\
\bar{z}(i-1, 2)\\
\ldotp \\
\ldotp \\
\ldotp \\
\bar{z}(i-1, k)\\
\end{array}\right\rbrack && \;\;i = 1,2,\ldots C \\ \\ &\{\bar{z}(0,1),\bar{z}(0,2),\ldots, \bar{z}(0,k)\} = \rho (X_0 - \bar{X}) \\
&\tag{D10} \label{D10} Q^{'} \bar{\gamma} \geq \tau_{L} \quad &&\\
& \tag{D11} \label{D11}\bar{\gamma}_j \geq 0, &&j = 1,2,\dots,m \\
\end{alignat*}
The initial condition $X_0 = \{x(0,1), x(0,2), \ldots, x(0,{\vert V \vert}), y(0,1), y(0,2), \ldots, y(0,{\vert V \vert})\}$ is used to obtain data points ($z(i,j)$) using ROM. $\bar{X}$ is determined as described in \cite{rathinam_petzold_2003}. The output from the dynamic optimization problem is $(\bar{z}^*, \bar{\gamma}^*)$. Then, $w_1 = \mathrm{diag(}W)\gamma^*$. Elements in $w_1$ less than $\epsilon_1$ are set to 0, and $w_1$ is updated. We update the solutions in the filtering step using this sparsified graph $G_1 = (V, E_1, w_1)$ when the uncertainty value in the covariance matrix of filtering exceeds a threshold.

\section{Experimental Results}
\textcolor{black}{
In this section, we present empirical results demonstrating the effectiveness of the proposed framework in reinforcing graph-based linear dynamical systems (\(\mathbb{D}\)) and nonlinear reaction-diffusion systems, specifically the chemical Brusselator model (\(\mathbb{C}\)), as described in Section \ref{spatio-temporal}. We evaluate the effectiveness of our framework by analyzing RMSE values under perturbations in initial conditions for surrogate models applied to both linear and reaction-diffusion systems on graphs. Additionally, we assess the impact of incorporating our framework on these models for both linear and non-linear dynamical systems. The influence of initial condition perturbations on neural ODE-based surrogate models, both with and without our framework discussed in Section~\ref{neuralodesl}, is also presented in Table~\ref{tab:rmse_comparison}. 
}\\
\begin{table}[htbp]
    \centering
    \renewcommand{\arraystretch}{1.2} 
    \setlength{\tabcolsep}{12pt} 
    \resizebox{\linewidth}{!}{ 
    \begin{tabular}{|l|c|c|c|}
        \hline
        \textbf{Experiment} & \textbf{Surrogate model} & \cellcolor{gray!20} \textcolor{black}{Surrogate model with Framework} & \textbf{ROM} \\
        & RMSE & \cellcolor{gray!20} RMSE & RMSE \\
        \hline
        Reaction-Diffusion with ROM surrogate model & 0.59 & \cellcolor{gray!20} \textbf{0.40} & — \\
        \hline
        Linear Diffusion with ROM surrogate model   & 5.17 & \cellcolor{gray!20} \textbf{0.38} & — \\
        \hline
        Linear Diffusion with Neural ODE surrogate model & 0.48 & \cellcolor{gray!20} \textbf{0.29} & 0.65 \\
        \hline
    \end{tabular}
    }
  \caption{Comparison of RMSE values across different experiments, analyzing the effects of perturbations of input on surrogate models. 
The \textbf{Reaction-Diffusion experiment} is conducted on a \textbf{40-node Erd\H{o}s-R$\acute{e}$nyi random graph} using the \textbf{chemical Brusselator dynamics} ($\mathbb{C}$), as detailed in Section~\ref{spatio-temporal}. 
The \textbf{Linear Diffusion experiment} is performed on a \textbf{30-node Erd\H{o}s-R$\acute{e}$nyi random graph} with the \textbf{graph-based dynamical system} ($\mathbb{D}$), employing the \textbf{ROM} (Section~\ref{POD-steps}) as the surrogate model. 
For both the \textbf{Linear Diffusion} and \textbf{Reaction-Diffusion} experiments, we assess the robustness of the surrogate model with \textbf{our framework} under perturbations. 
In the \textbf{Linear Diffusion experiment with neural ODEs} (Section~\ref{neuralodesl}), we analyze the impact of perturbations on a \textbf{neural-ODE surrogate model} for the diffusion equation ($\mathbb{D}$ in Section~\ref{spatio-temporal}) on a \textbf{10-node complete graph}. 
Surrogate model with our framework achieves the \textbf{lowest RMSE} across all experiments, demonstrating its effectiveness in improving surrogate model accuracy under perturbations.}

    \label{tab:rmse_comparison}
\end{table}

\textbf{Remark:} \textcolor{black}{In certain graph structures, empirical experiments revealed instances of particle filter divergence.} Particle filter divergence is a critical issue that must be carefully addressed, as it compromises estimation accuracy and reduces the framework's effectiveness. Several factors can contribute to this phenomenon, including suboptimal filter tuning, modeling inaccuracies, inconsistent measurement data, or system-related issues such as hardware-induced delays. Specifically, inaccurate likelihood estimations due to imprecise noise assumptions, erroneous process models, or delayed measurement updates can lead to divergence. For further examples, see \cite{elfring2021particle}. Empirically, we observed that in certain graph configurations, the particle filter exhibited divergence, necessitating additional updates in the filtering step.
\label{numerical_results}
\subsection{Linear dynamical system represented on graphs}
We present the experimental results for the linear diffusion equation (\(\mathbb{D}\)) on a 30-node Erd\H{o}s-R$\acute{e}$nyi random graph using the 4-point orthogonal collocation method with 20 elements. The parameters used in Algorithm \ref{algo_degree} and Section \ref{dynamiclinear} are as follows: \(\epsilon = 0.5\), \(T=0.15\), \(k = \mathrm{min}(\lceil \frac{n}{5} \rceil, 50)\), with two trajectories considered. The number of clusters is set to \(p=30\), while the particle filter employs 15,000 particles with \(\epsilon_1 = 1 \times 10^{-5}\). The noise terms \(w_k\) and \(\mu_k\) are assumed to be normally distributed with zero mean and variances \(\zeta_x I\) and \(\zeta_y I\), where \(\zeta_x = 0.01\) and \(\zeta_y = 1 \times 10^{-7}\).

Following the dynamic optimization step, the resulting graph exhibited 31 sparse edges, a substantial reduction from the original graph's 336 edges. During the filtering step, 193 updates were required for prediction over 1000 timesteps.  \textcolor{black}{Figure \ref{linearplot} compares the reduced order model solution, the actual solution, and the reduced-order model solution with our framework for the linear dynamical system ($\mathbb{D}$) described in Section \ref{spatio-temporal}. Figure~\hyperref[linearplot]{3(c)}
 shows that the ROM solution with our framework closely resembles the actual solution in Figure~\hyperref[linearplot]{3(b)} than the reduced-order model solution in Figure~\hyperref[linearplot]{3(a)}, as discussed in Section \ref{POD-steps}.} 
\begin{figure}[htbp]
  \centering
  \includegraphics[width=\linewidth]{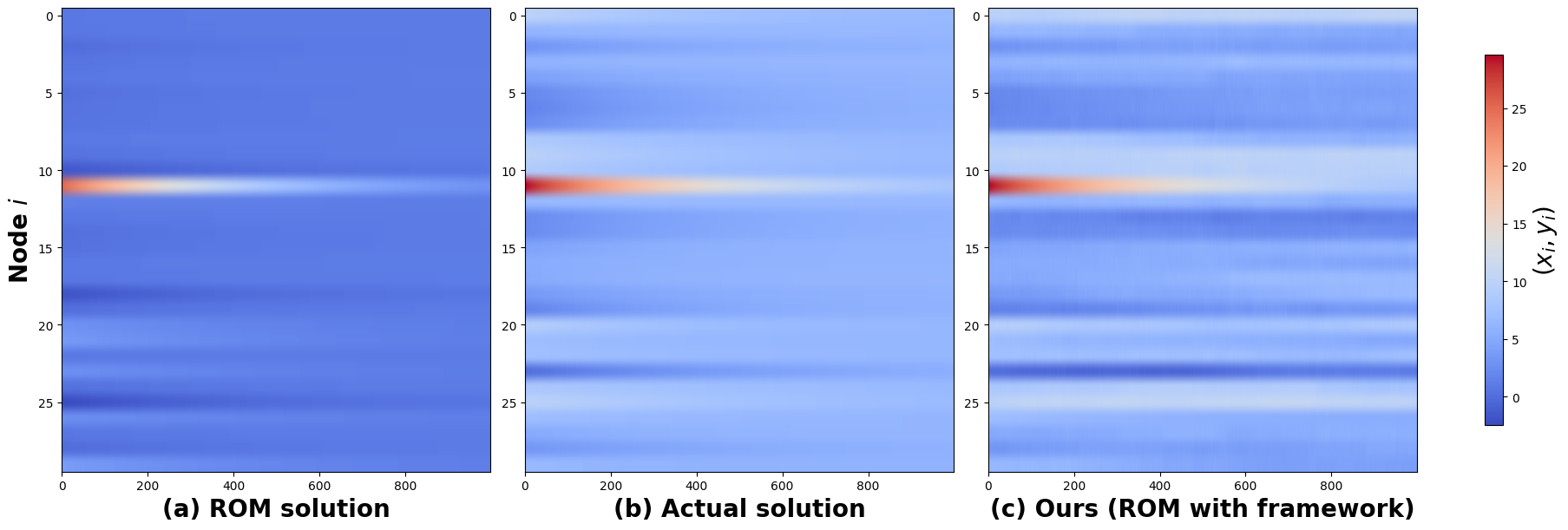}
  \caption{Comparison of the (a) ROM Solution, (b) Actual Solution, and (c) ROM solution with our framework for the graph-based dynamical system $\mathbb{D}$, as described in Section \ref{spatio-temporal}.}
  \label{linearplot}
\end{figure}

\subsection{Reaction-diffusion system represented on graphs}
\begin{figure}[H]
  \centering
  \includegraphics[width=\textwidth]{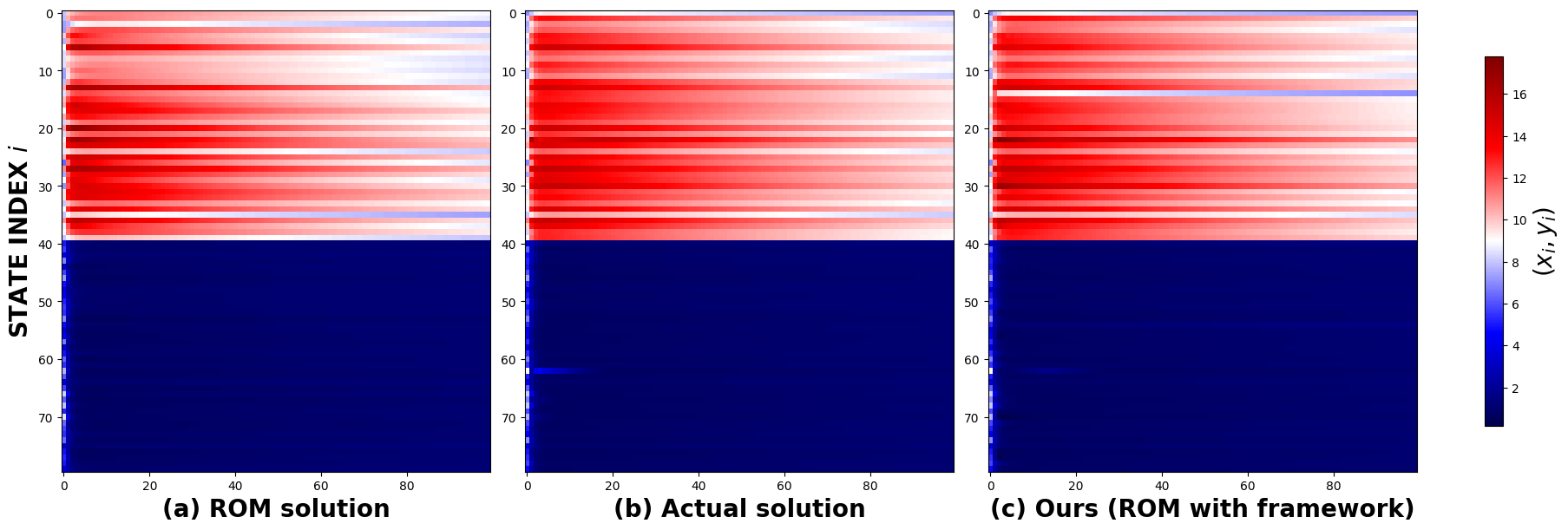}
  \caption{Comparing (a) ROM, (b) Actual, and (c) ROM solution with our framework for the case of the chemical Brusselator reaction-diffusion system (Eq. \ref{Rd dynamics}).}
  \label{rd_fig}
\end{figure}


\par The results for the nonlinear case are based on the alternating self-dynamics Brusselator model applied to a 40-node Erd\H{o}s-R$\acute{e}$nyi random graph, using the 4-point orthogonal collocation method. The number of trajectories is set to two, and the particle filter employs 15,000 particles. The total simulation time is set to \( T = 3 \), with the number of clusters fixed at \( p = 30 \). The sparsification parameter $\epsilon_1$ is set to  $10^{-5}$, and the minimum connectivity constraint is defined as \( \tau_L = 0.1 \mathrm{d}_s \), where \( \mathrm{d}_s \) denotes the minimum degree of the graph. The reduced-order dimension is given by \( k = \mathrm{min}(\lceil \frac{2n}{5} \rceil, 50) \). The noise terms \( w_k \) and \( \mu_k \) are assumed to be normally distributed with zero mean and variances \( \zeta_x I \) and \( \zeta_y I \), where \( \zeta_x = 0.01 \) and \( \zeta_y = 1 \times 10^{-7} \).
The graph obtained from the dynamic optimization step contained 286 sparse edges, a reduction from the original graph's 336 edges. The filtering step required a total of 50 updates for predictions over 100 timesteps. 

\textcolor{black}{Figure \ref{rd_fig} presents a comparison of the actual solution of the chemical Brusselator reaction-diffusion system (\(\mathbb{C}\) in Section \ref{spatio-temporal}) for a randomly initialized condition, with both the reduced-order model (Section \ref{POD-steps}) and the ROM solution with our framework. The grid-based representation highlights regions marked in black, indicating areas where the absolute error of the particle filter solution is lower than that of the ROM solution.}

\section{Benchmarking the framework using neural ODEs}
\label{benchmark}
\textcolor{black}{To broaden the applicability of our framework, we evaluate its effectiveness using a neural-ODE-based surrogate model, where state vectors are observed at discrete time intervals.}
Neural ordinary differential equations (neural ODEs) provide a powerful framework for modeling and analyzing complex dynamical systems~\cite{DBLP:journals/corr/abs-1005-0265, yan2022robustness}. They offer a flexible approach to capturing system dynamics by integrating the principles of ordinary differential equations with neural networks. Unlike traditional neural networks that process inputs through discrete layers, neural ODEs represent system dynamics continuously using an ordinary differential equation to govern the evolution of hidden states over time. Neural ODEs learn system dynamics by learning the parameters of the differential equation using the adjoint sensitivity method (\cite{lewis_lakshmivarahan_dhall_2009}). This approach allows the model to capture complex temporal dependencies in data. By leveraging the continuous-time nature of differential equations, neural ODEs provide key advantages, including the ability to model irregularly sampled time-series data and accommodate variable-length inputs.

 \par To illustrate the framework, we consider a linear dynamical system ($\mathbb{D}$) from Section~\ref{spatio-temporal}, defined as follows:
\begin{equation}
 \frac{dx}{dt} = -Lx.
 \label{nodesystem}
\end{equation}

\textcolor{black}{The key distinction in applying our framework with the neural ODE model as the surrogate is the absence of the dynamic optimization step outlined in Section \ref{framework}. Instead, an additional step is included to train the parameters of the neural ODE surrogate model, as described in Equation \ref{costfuncnode}. The neural ODE solutions are treated as observations in the filtering step after projection ($z_k$), where the POD method is applied for dimensionality reduction. The forward state dynamics and the state observation relationship, as applied in Step 4 of the filtering process within the framework, are given by:}
\begin{equation}
\begin{cases}
x_{k+1} = M({x}_k) + A_{k+1} v_{k+1} + w_k, \\
z_k =\rho \left(x_k -\bar{x} \right)+ R_k \beta_k + \mu_k.
\end{cases}
\end{equation}
\\
When utilizing the neural ODE surrogate model with an Euler discretization step \(\Delta t_k\), the forward model is expressed as:
\begin{equation*}
M({x}_k) = x_k + \Delta t_k \cdot nn(x_k).
\end{equation*}
\\
At iteration $k$, the vector $\beta_k$ is updated by solving the following linear system:
\begin{equation*}
R_{k+1} \beta_k = z_{k+1} - \rho(x_{k+1} - \bar{x}).
\end{equation*}
\\
The matrices $A_{k+1}$ and $R_{k+1}$ are computed as described in Step 3 of the framework (Section~\ref{framework}). The solution at time step \( k+1 \) is determined as:
\begin{equation*}
x_{k+1} = x_k + \Delta t_k  nn(x_k),
\end{equation*}
where \( \Delta t_k = t_{k+1} - t_k \). The vector $v_k$ is initialized as $\mathbf{0}_{n \times 1}$. The neural network architecture is structured as follows:
\begin{equation*}
nn(x(t)) = \sinh{ (\theta_3 \theta_2 \theta_1 x(t) + \theta_3 \theta_2 b_1 + b_3)}.    
\label{neuralodesl}
\end{equation*}

\begin{align*}
\theta_1 &\in \mathbb{R}^{10 \times 50}, \theta_2 \in \mathbb{R}^{50 \times 50}, \theta_3 \in \mathbb{R}^{50 \times 10}, b_1 \in \mathbb{R}^{50}, b_3 \in \mathbb{R}^{10}.
\end{align*}

The experiment involved 41 updates, with predictions performed over 100 timesteps. The number of particles in the particle filter step was set to 20,000. The noise terms \( w_k \) and \( \mu_k \) were assumed to follow normal distributions with zero mean and variances \( \zeta_x I \) and \( \zeta_y I \), where \( \zeta_x = 0.01 \) and \( \zeta_y = 1 \times 10^{-3} \). The Laplacian matrix \( L \) was chosen as the Laplacian of the complete graph \( \mathbb{K}_{10} \), and \( N = 30 \) observations were randomly sampled over the time interval \( t = 0 \) to \( t = 0.05 \).

The cost function \(J\) (Eq. \ref{costfuncnode}) for estimating the parameters of the neural ODE model is given by: 

\begin{equation}
 J(\theta) = \frac{1}{2} \sum_{k=1}^N \Vert x_k - \tilde{x}_k \Vert_2^2 + \alpha_1 \Vert \theta \Vert_1.   
\label{costfuncnode}
\end{equation}
The parameter vector \( \theta \), obtained by flattening \( \theta_1, \theta_2, \theta_3, b_1, \) and \( b_3 \), is represented as \( \theta \in \mathbb{R}^{3560 \times 1} \). \textcolor{black}{For experimentation, the parameters in Eq.~\ref{costfuncnode} are determined using the adjoint sensitivity method with an Euler discretization scheme.} The number of clusters is fixed at 40.
 
\begin{figure}[H]
  \centering
  \includegraphics[width=\textwidth]{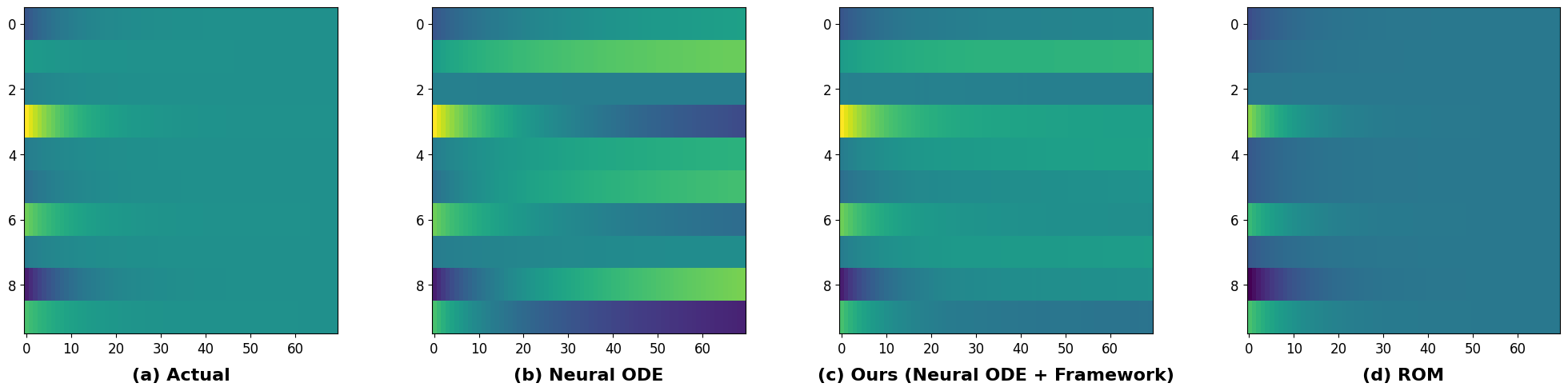}
  
  \caption{Comparison of (a) the actual solution, (b) the neural-ODE solution, (c) the neural-ODE solution with our framework, and (d) the ROM solution for a random initial condition in the case of a linear dynamical system on graphs ($\mathbb{D}$ in Section \ref{spatio-temporal}).}
  \label{allcomparenode}
\end{figure}

  
\begin{figure}
\begin{subfigure}{0.45\textwidth}
    \includegraphics[width=\linewidth, height = 5cm]{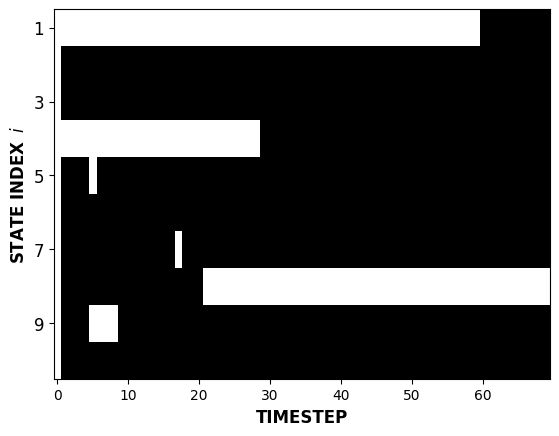}
    \caption{Comparing absolute value of errors of neural-ODE solution with our framework and the neural ODE solution. Here the black regions indicate grid points where the absolute value of error in the neural-ODE solution with our framework is less than neural ODE solution.}
    \label{cmpnode1}
  \end{subfigure}
\hfill
\begin{subfigure}{0.45\textwidth}
    \includegraphics[width=\linewidth, height = 5cm]{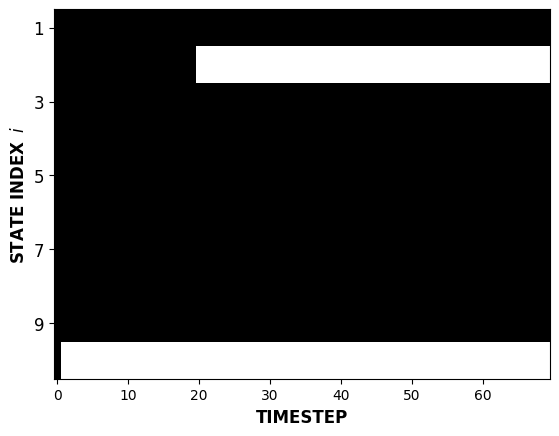}
    \caption{Comparing absolute value of errors of neural-ODE solution with our framework and the ROM solution. Here the black regions indicate grid points where the absolute value of error in the neural-ODE solution with our framework is less than ROM solution.}
    \label{cmpnode2}
  \end{subfigure}

  \caption{Comparison of absolute errors between the neural-ODE surrogate model with our framework, the standalone neural-ODE model, and the ROM solution. Black regions indicate grid points where the solutions given by the framework achieves a lower absolute error than the other methods.}
  \label{comparenode}
\end{figure}

\textcolor{black}{Figure~\ref{allcomparenode} compares the neural-ODE solution for a random initial condition with the ROM and neural-ODE solution with our framework. The neural-ODE solution (Figure~\hyperref[allcomparenode]{5(b)}) exhibits higher noise levels compared to the actual solution (Figure~\hyperref[allcomparenode]{5(a)}). Figure~\ref{comparenode} contrasts the ROM solution (Figure~\hyperref[allcomparenode]{5(d)}) of the linear dynamical system (Equation~\ref{nodesystem}) with the neural-ODE solution with our framework (Figure~\hyperref[allcomparenode]{5(c)}). In the grid-wise comparison of solutions, the black regions in Figures~\ref{cmpnode1} and \ref{cmpnode2} indicate areas where the error in the neural-ODE solution with our framework is lower than that in the neural-ODE and ROM solutions. Notably, the neural-ODE solution with our framework exhibits a greater number of black regions, suggesting a closer resemblance to the actual solution compared to the ROM and neural ODE solutions.
}

 \section{Conclusion and Discussion}
\textcolor{black}{
In this study, we proposed a novel framework to improve the robustness of surrogate models against input perturbations, addressing a key challenge in modeling complex systems while ensuring computational efficiency. By integrating machine learning techniques with stochastic filtering, our approach has demonstrated significant improvements in surrogate model accuracy under perturbations. The experimental results presented in Section~\ref{numerical_results} provide strong validation of the framework’s effectiveness. } 

The versatility of our framework is evident from its application to dynamical systems represented on graphs and its extension to a general setting, where a neural-ODE-based surrogate model was employed to simulate complex physical phenomena. This not only highlights the robustness of our approach but also underscores its applicability across diverse domains, as discussed in Section~\ref{benchmark}.  

Despite its strengths, our framework has certain limitations. Notably, its reliance on stochastic filtering methods introduces computational overhead and may also be susceptible to filter divergence in certain cases. While our study focuses on undirected graphs, many real-world complex systems involve directed networks, such as those governed by the complex Ginzburg-Landau equation~\cite{cencetti_clusella_fanelli_2018}. Addressing directed graph structures presents an important avenue for future research, with potential extensions of our framework to handle these more intricate dynamical systems.

\section*{Acknowledgments}
We thank Prof. S. Lakshmivarahan and Prof. Arun Tangirala for their insightful feedback, contributing to the enhancement of this work. \textcolor{black}{We thank the anonymous reviewers for their insightful comments, which significantly improved the quality of the manuscript.} This work was partially supported by the MATRIX grant MTR/2020/000186 of the Science and Engineering Research Board of India.
\bibliographystyle{unsrt} 
\bibliography{bibfile} 

\begin{thebibliography}{10}

\bibitem{Brockmann_Helbing_2013}
Dirk Brockmann and Dirk Helbing.
\newblock The hidden geometry of complex, network-driven contagion phenomena.
\newblock {\em Science}, 342(6164):1337–1342, 2013.

\bibitem{Maslov_Ispolatov_2007}
Sergei Maslov and I.~Ispolatov.
\newblock Propagation of large concentration changes in reversible
  protein-binding networks.
\newblock {\em Proceedings of the National Academy of Sciences},
  104(34):13655–13660, 2007.

\bibitem{rathinam_petzold_2003}
Muruhan Rathinam and Linda~R. Petzold.
\newblock A new look at proper orthogonal decomposition.
\newblock {\em SIAM Journal on Numerical Analysis}, 41(5):1893–1925, 2003.

\bibitem{articlePODfield1}
Suhan Kim and Hyunseong Shin.
\newblock Data-driven multiscale finite-element method using deep neural
  network combined with proper orthogonal decomposition.
\newblock {\em Engineering with Computers}, 40, 04 2023.

\bibitem{articlePODfield2}
B.A. Le, Julien Yvonnet, and Qi-Chang He.
\newblock {Computational homogenization of nonlinear elastic materials using
  neural networks}.
\newblock {\em {International Journal for Numerical Methods in Engineering}},
  104(12):1061--1084, 2015.

\bibitem{articlePODfield3}
Yanwen Huang and Yuanchang Deng.
\newblock A hybrid model utilizing principal component analysis and artificial
  neural networks for driving drowsiness detection.
\newblock {\em Applied Sciences}, 12(12), 2022.

\bibitem{Lguensat_Tandeo_Ailliot_Pulido_Fablet_2017}
Redouane Lguensat, Pierre Tandeo, Pierre Ailliot, Manuel Pulido, and Ronan
  Fablet.
\newblock The analog data assimilation.
\newblock {\em Monthly Weather Review}, 145(10):4093–4107, 2017.

\bibitem{Park}
D.C. Park and Yan Zhu.
\newblock Bilinear recurrent neural network.
\newblock {\em Proceedings of 1994 IEEE International Conference on Neural
  Networks (ICNN’94)}.

\bibitem{park2010time}
D.-C. Park.
\newblock A time series data prediction scheme using bilinear recurrent neural
  network.
\newblock In {\em 2010 International Conference on Information Science and
  Applications}, pages 1--7, Seoul, Korea (South), 2010. IEEE.

\bibitem{brajard2020combining}
Julien Brajard, Alberto Carrassi, Marc Bocquet, and Laurent Bertino.
\newblock Combining data assimilation and machine learning to emulate a
  dynamical model from sparse and noisy observations: A case study with the
  lorenz 96 model.
\newblock {\em Journal of Computational Science}, 44:101171, 2020.

\bibitem{dueben2018challenges}
Peter~D. Dueben and Peter Bauer.
\newblock Challenges and design choices for global weather and climate models
  based on machine learning.
\newblock {\em Geoscientific Model Development}, 11(10):3999--4009, 2018.

\bibitem{fablet2018bilinear}
Ronan Fablet, Souhaib Ouala, and Cédric Herzet.
\newblock Bilinear residual neural network for the identification and
  forecasting of geophysical dynamics.
\newblock In {\em 2018 26th European Signal Processing Conference (EUSIPCO)},
  pages 1477--1481, Rome, 2018. IEEE.

\bibitem{long2018pde}
Zichao Long, Yiping Lu, Xianzhong Ma, and Bin Dong.
\newblock {PDE-net}: Learning {PDEs} from data.
\newblock In {\em Proceedings of the 35th International Conference on Machine
  Learning}, page~9, 2018.

\bibitem{bocquet2019data}
Marc Bocquet, J{\'e}r{\'e}mie Brajard, Alberto Carrassi, and Laurent Bertino.
\newblock Data assimilation as a learning tool to infer ordinary differential
  equation representations of dynamical models.
\newblock {\em Nonlinear Processes in Geophysics}, 26(3):143--162, 2019.

\bibitem{bocquet2020bayesian}
Marc Bocquet, J{\'e}r{\'e}mie Brajard, Alberto Carrassi, and Laurent Bertino.
\newblock Bayesian inference of chaotic dynamics by merging data assimilation,
  machine learning and expectation-maximization.
\newblock {\em Foundations of Data Science}, 2(1):55--80, 2020.

\bibitem{sakov2018iterative}
Pavel Sakov, Jean-Michel Haussaire, and Marc Bocquet.
\newblock An iterative ensemble kalman filter in the presence of additive model
  error.
\newblock {\em Quarterly Journal of the Royal Meteorological Society},
  144(713):1297--1309, 2018.

\bibitem{Farchi_2021}
Alban Farchi, Patrick Laloyaux, Massimo Bonavita, and Marc Bocquet.
\newblock Using machine learning to correct model error in data assimilation
  and forecast applications.
\newblock {\em Quarterly Journal of the Royal Meteorological Society},
  147(739):3067--3084, jul 2021.

\bibitem{HEDENGREN2014133}
John~D. Hedengren, Reza~Asgharzadeh Shishavan, Kody~M. Powell, and Thomas~F.
  Edgar.
\newblock Nonlinear modeling, estimation and predictive control in apmonitor.
\newblock {\em Computers \& Chemical Engineering}, 70:133--148, 2014.
\newblock Manfred Morari Special Issue.

\bibitem{nlp_siam}
Lorenz~T. Biegler.
\newblock {\em Nonlinear Programming}.
\newblock Society for Industrial and Applied Mathematics, 2010.

\bibitem{LANDSBERG1972}
P.~T. Landsberg.
\newblock The fourth law of thermodynamics.
\newblock {\em Nature}, 238(5361):229--231, 1972.

\bibitem{cencetti_clusella_fanelli_2018}
Giulia Cencetti, Pau Clusella, and Duccio Fanelli.
\newblock Pattern invariance for reaction-diffusion systems on complex
  networks.
\newblock {\em Scientific Reports}, 8(1), 2018.

\bibitem{Hens2019}
Chittaranjan Hens, Uzi Harush, Simi Haber, Reuven Cohen, and Baruch Barzel.
\newblock Spatiotemporal signal propagation in complex networks.
\newblock {\em Nature Physics}, 15(4):403--412, 2019.

\bibitem{lewis_lakshmivarahan_dhall_2009}
John~M. Lewis, Sivaramakrishnan Lakshmivarahan, and Sudarshan~Kumar Dhall.
\newblock {\em Dynamic Data Assimilation: A least squares approach}.
\newblock Cambridge Univ. Press, 2009.

\bibitem{resampling_matlab}
Tiancheng Li, Miodrag Bolic, and Petar~M. Djuric.
\newblock Resampling methods for particle filtering: Classification,
  implementation, and strategies.
\newblock {\em IEEE Signal Processing Magazine}, 32(3):70--86, 2015.

\bibitem{Chung1997}
Fan R.~K. Chung.
\newblock {\em Spectral Graph Theory}.
\newblock American Mathematical Society, Providence, RI, 1997.

\bibitem{burden_faires_burden_2016}
Richard~L. Burden, J.~Douglas Faires, and Annette~M. Burden.
\newblock {\em Numerical analysis}.
\newblock Cengage Learning, 2016.

\bibitem{lai2021graph}
Ming-Jun Lai, Jiaxin Xie, and Zhiqiang Xu.
\newblock Graph sparsification by universal greedy algorithms, 2021.
\newblock \url{https://doi.org/10.48550/arXiv.2007.07161}.

\bibitem{spielman_teng_2011}
Daniel~A. Spielman and Shang-Hua Teng.
\newblock Spectral sparsification of graphs.
\newblock {\em SIAM Journal on Computing}, 40(4):981–1025, 2011.

\bibitem{spielman2009graph}
Daniel~A. Spielman and Nikhil Srivastava.
\newblock Graph sparsification by effective resistances, 2009.

\bibitem{boucheron2013concentration}
St$\acute{o}$phane Boucheron, G$\acute{a}$bor Lugosi, and Olivier Bousquet.
\newblock Concentration inequalities.
\newblock {\em St$\acute{o}$phane Boucheron, G$\acute{a}$bor Lugosi, and
  Olivier Bousquet}, 2013.

\bibitem{farber2011upper}
Miriam Farber and Ido Kaminer.
\newblock Upper bound for the laplacian eigenvalues of a graph, 2011.
\newblock \url{https://doi.org/10.48550/arXiv.1106.0769}.

\bibitem{book}
Martin Grötschel, Sven Krumke, and Jörg Rambau.
\newblock {\em Online Optimization of Large Scale Systems}.
\newblock 01 2001.

\bibitem{bartusiak}
R.~Donald Bartusiak.
\newblock Nlmpc: A platform for optimal control of feed- or product-flexible
  manufacturing.
\newblock {\em Assessment and Future Directions of Nonlinear Model Predictive
  Control}, page 367–381.

\bibitem{Nagy2007}
Zoltan~K. Nagy, Bernd Mahn, R{\"u}diger Franke, and Frank Allg{\"o}wer.
\newblock {\em Real-Time Implementation of Nonlinear Model Predictive Control
  of Batch Processes in an Industrial Framework}, pages 465--472.
\newblock Springer Berlin Heidelberg, Berlin, Heidelberg, 2007.

\bibitem{ajayakumar2023data}
Abhishek Ajayakumar and Soumyendu Raha.
\newblock Data assimilation for sparsification of reaction diffusion systems in
  a complex network, 2023.
\newblock \url{https://doi.org/10.48550/arXiv.2303.11943}.

\bibitem{nocedalbook}
Jorge Nocedal and Stephen~J. Wright.
\newblock {\em Numerical optimization}.
\newblock Springer, 2006.

\bibitem{luenberger2021}
David~G. Luenberger and Yinyu Ye.
\newblock {\em Linear and nonlinear programming}.
\newblock Springer, 2021.

\bibitem{elfring2021particle}
Jeroen Elfring, Elena Torta, and Rob van~de Molengraft.
\newblock Particle filters: A hands-on tutorial.
\newblock {\em Sensors (Basel, Switzerland)}, 21(2):438, 2021.

\bibitem{DBLP:journals/corr/abs-1005-0265}
Wai~Shing Fung and Nicholas J.~A. Harvey.
\newblock Graph sparsification by edge-connectivity and random spanning trees.
\newblock {\em CoRR}, abs/1005.0265, 2010.

\bibitem{yan2022robustness}
Hanshu Yan, Jiawei Du, Vincent Y.~F. Tan, and Jiashi Feng.
\newblock On robustness of neural ordinary differential equations, 2022.

\end{thebibliography}

\end{document}